\documentclass[prb,aps,floatfix,amsmath,amssymb,superscriptaddress,tightenlines,twocolumn]{revtex4}
\usepackage{graphicx}% Include figure files
\usepackage{epstopdf}
\newcommand{\be}{\begin{equation}}
\newcommand{\ee}{\end{equation}}
\usepackage{amsmath}
\usepackage{amsfonts}
\usepackage{bm}
\usepackage{color}
\usepackage{subfigure}
\usepackage{amsthm}
\usepackage[latin1]{inputenc}

\newcommand{\bra}[1]{\left\langle #1 \right|}
\newcommand{\ket}[1]{\left|#1\right\rangle}

\newcommand{\Tr}{\textrm{Tr}}

\newtheorem{lem}{Lemma}
\newtheorem{thm}{Theorem}
\newtheorem{defi}{Definition}

\begin{document}
\title{Determining the structure of real-space entanglement spectrum from approximate conditional independence}
\author{Isaac H. Kim}
\affiliation{Institute of Quantum Information, California Institute of Technology, Pasadena CA 91125, USA}

\date{\today}
\begin{abstract}
We study the ground state of a gapped quantum many-body system whose entanglement entropy $S_A$ can be expressed as $S_A = a|\partial A| - \gamma$, where $a, \gamma$ are some constants and $|\partial A|$ is an area of the subsystem $A$. By using a recently proved operator extension of strong subadditivity of entropy,[I. H. Kim, J. Math. Phys. 53, 122204 (2012)] we show that certain linear combination of the real-space entanglement spectrum has a small correlation with almost any local operator. Our result implies that there exists a structure relating the real-space entanglement spectrum over different subsystems. Further, this structure is inherited from the generic property of the ground state alone, suggesting that the locality of the entanglement spectrum may be attributed to the area law of entanglement entropy.
\end{abstract}

\maketitle
\section{Introduction}
It is commonly believed that gapped phases of quantum many-body systems exhibit area law: entanglement entropy of a simply connected subsystem increases as the area of the boundary.\cite{Eisert2008} Overwhelming amount of evidences supporting this statement has been suggested, including the explicit proof for a ground state of 1D gapped system\cite{Hastings2007a}, exactly solvable models\cite{Levin2005}, and variational wavefunctions\cite{Verstraete2006}. Constant subcorrection to the entanglement entropy - also known as the topological entanglement entropy - can be extracted by judiciously choosing a set of subsystems that cancel out the boundary contributions.\cite{Kitaev2006,Levin2006} Topological entanglement entropy is believed to be a universal constant characterizing the phase of the quantum many-body system.

Li and Haldane(LH) were the first to realize that the spectrum of the reduced density matrix may reveal an information about the phase that cannot be inferred from the entanglement entropy alone.\cite{Li2008,Thomale2010} While LH studied reduced density matrix in the orbital cuts, one may study its spectrum along a real-space partition and arrive at a similar conclusion.\cite{Sterdyniak2011,Dubail2011,Rodriguez2012} In particular, it has been recently suggested by several authors that entanglement spectrum along a real-space partition has a low-lying part that can be described by a local field theory.\cite{Cirac2011,Dubail2012}

Topological entanglement entropy can be obtained from a real-space entanglement spectrum of variational wavefunctions, similar to the way it is  extracted from the entanglement entropy.\cite{Dubail2012} Consequently, the corresponding linear combination of entanglement spectrum is ``topological'', in a sense that i) it does not interact with any local observable ii) it is equal to the topological entanglement entropy.

Here we claim that the existence of such topological operator can be attributed to an approximate conditional independence of these quantum states. A tripartite state $\rho_{ABC}$ is conditionally independent if conditional mutual information $I(A:C|B)  = S_{AB} + S_{BC} - S_{B} - S_{ABC}$ is equal to $0$. A state is approximately conditionally independent if $0$ is replaced by a small number $\epsilon>0$. To the best of author's knowledge, Hastings and Poulin were the first to point out that there can be configurations that are conditionally independent even in a quantum many-body system with long range entanglement.\cite{Poulin2010a} To illustrate their idea, suppose entanglement entropy satisfies an area law with a \emph{universal} constant subcorrection term.
\begin{equation}
S_A = a|\partial A| - \gamma, \label{eq:area_law}
\end{equation}
One can show that $I(A:C|B)=0$ for a choice of $A,B,C$ such that i) $AB,BC,B,ABC$ are all simply connected ii) $A$ and $C$ do not share a boundary.

A state that is conditionally independent saturates the equality condition of the strong subadditivity of entropy.\cite{Lieb1973} Such state forms a quantum Markov chain, and the structure of the reduced density matrix is vastly restricted compared to an arbitrary state.\cite{Petz2003,Hayden2004,Leifer2007} It is important to note that one cannot directly use these results for a generic quantum many-body system, since the conditional independence condition is unlikely to hold exactly. Still, one may hope these properties to hold approximately for a sufficiently small conditional mutual information. This is precisely the key idea behind this paper. More specifically, we shall use the recently discovered operator extension of the strong subadditivity of entropy as our main technical tool.\cite{Kim2012a}

Rest of the paper is structured as follows. In Section \ref{section:ACI}, we shall briefly review several information-theoretic inequalities. In Section \ref{section:correlation_bound}, we shall introduce a diagrammatic trick that leads to the main result of this paper. Its physical interpretation shall be given in Section \ref{section:interpretation}.

\section{Approximately conditionally independent states\label{section:ACI}}
Strong subadditivity of entropy is one of the most widely used tools in quantum information theory. Its importance stems from the fact that there exists a variety of nontrivial structure theorems that relate the reduced density matrix of different subsystems if the inequality is saturated with an equality condition.\cite{Petz2003,Hayden2004,Leifer2007} In particular, Petz showed that the following relation holds if and only if the conditional mutual information $I(A:C|B)$ is equal to $0$.\cite{Petz2003}
\begin{equation}
\hat{H}_{AB} + \hat{H}_{BC} - \hat{H}_B - \hat{H}_{ABC}=0,\label{eq:equality_condition}
\end{equation}
where $\hat{H}_A = -I_{A^c} \otimes \log \rho_A$ is a formal definition of the entanglement spectrum. From now on, we denote the left hand side of the equation as $\hat{H}_{A:C|B}$ and refer to it as a \emph{conditional mutual spectrum} of $ABC$. It follows that
\begin{equation}
\mathcal{C}(\hat{H}_{A:C|B} , X) = 0, \label{eq:ccf}
\end{equation}
where $\mathcal{C}(\hat{H}_{A:C|B} , X) = \langle \hat{H}_{A:C|B} X\rangle - \langle \hat{H}_{A:C|B}\rangle \langle X\rangle$ is a connected correlation function between the conditional mutual spectrum and an arbitrary operator $X$. $\langle \cdots \rangle$ denotes ground state expectation value.

While such operator trivially has zero correlation with any local operator, exact conditional independence is rarely satisfied by any realistic physical systems.  Motivated by this observation, author has recently obtained a nontrivial statement about the spectrum of the reduced density matrices.\cite{Kim2012a}
\begin{thm}
\begin{equation}
\Tr_{BC}(\rho_{ABC} \hat{H}_{A:C|B}) \geq 0. \label{eq:OSSA}
\end{equation}
\end{thm}
We would like to emphasize two important facts about this inequality. First, Eq.\ref{eq:OSSA} reproduces a statement similar to Eq.\ref{eq:equality_condition} when the conditional mutual information is $0$. This can be seen from the following lemma.
\begin{lem}
\begin{equation}
|\Tr_{ABC}(\rho_{ABC} \hat{H}_{A:C|B} O_A)| \leq \|O_A \| I(A:C|B),
\end{equation}
where $\| \cdots \|$ is $l_{\infty}$ norm. \label{lemma:lemma1}
\end{lem}
\begin{proof}
Let $\mathcal{I}_A=\Tr_{BC}(\rho_{ABC} \hat{H}_{A:C|B})$. $|\Tr_A (\mathcal{I}_AO_A) | \leq \|O_A \| |\mathcal{I}_A|_1$. $|\cdots|_1$ is the $l_1$ norm. Since $\mathcal{I}_A$ is positive, $|\mathcal{I}_A|_1 = \Tr_A \mathcal{I}_A = I(A:C|B)$.
\end{proof}
If the conditional mutual information vanishes, the corresponding conditional mutual spectrum has zero correlation with any operator supported on $A$. Furthermore, since both $\hat{H}_{A:C|B}$ and $I(A:C|B)$ are symmetric under the exchange of $A$ and $C$, the same statement holds for an operator supported on $C$ as well. Secondly, Eq.\ref{eq:OSSA} is satisfied by any quantum states. Therefore, unlike Eq.\ref{eq:equality_condition}, it can be applied to quantum states that  \emph{approximately} saturate the strong subadditivity of entropy.

\section{Correlation bound for entanglement spectrum\label{section:correlation_bound}}
The main goal of this section is to obtain a statement that resembles Eq.\ref{eq:ccf} when the global state is a ground state of a gapped quantum many-body system. Such correlation bound can be easily obtained in certain cases using Lemma \ref{lemma:lemma1} alone, but there are also important caveats. For example, there are choices of subsystems that yield a nonzero value of conditional mutual information even at a fixed point of some renormalization-group flow.\cite{Kitaev2006,Levin2006} Furthermore, Lemma \ref{lemma:lemma1} alone cannot produce any bound on the correlation between the conditional mutual spectrum $\hat{H}_{A:C|B}$ and an operator supported on $B$. We shall show that, despite these shortcomings, it is still possible to obtain a bound analogous to Eq.\ref{eq:ccf} under a reasonable set of assumptions.

A brief comment on the notation is in order. For a conditional mutual spectrum $\hat{H}_{A:C|B}$, we shall refer $B$ as the {\it reference party} and $A,C$ as {\it target parties}. Also, we shall diagrammatically represent the operator $\hat{H}_{A:C|B}$ with the following rule. The reference party corresponds to the region with a `R' sign. Each of the target parties corresponds to one of the simply connected regions with a `T' sign. When taking a partial trace, subsystem $X$ is used to denote the nontrivial support of operator $X$. Shaded region in the diagram is a nontrivial support of $X$.

We postulate the following modified formula for the entanglement entropy to account for the deviations from the ideal area law.
\begin{equation}
S_A = a|\partial A| - \gamma + \epsilon_A.
\end{equation}
\begin{equation}
S_A + S_B - S_{AB} = \epsilon_{A:B}.
\end{equation}
For a large enough subsystem size, we expect $\epsilon_A$ to approach $0$. $\epsilon_{A:B}$  denotes a long range correlation of the ground state. Due to the exponential clustering theorem, we expect $\epsilon_{A:B}$ to scale as $\min(|A|,|B|)^2 e^{-\frac{2l}{ \xi}}$, where $\xi$ is the  correlation length and $|A|$ is the volume of the subsystem $A$.\footnote{The volume factor was chosen in such a way that the bound on connected correlation function from mutual information in Reference \cite{Wolf2007} yields the exponential clustering theorem in Reference \cite{Nachtergaele2006}.}

To simplify the analysis, we assume that each of the subsystems are sufficiently smooth and their boundary lengths are $O(l)$. We assume that the support of $X$ is sufficiently small compared to the size of the subsystems. We also assume that $X$ is supported on only one of the subsystems that partitions the system.

\subsection{Deformation moves}
The key idea for generalizing Eq.\ref{eq:ccf} is that one can decompose $\hat{H}_{A:C|B}$ into a sum of $\hat{H}_{A_i:C_i|B_i}$ in such a way that either i) $I(A_i:C_i|B_i)$ is small or ii) $A_iB_iC_i$ is sufficiently far away from the support of $X$. Such decomposition can be derived from a simple application of the following \emph{chain rule of conditional mutual spectrum}.
\begin{equation}
\hat{H}_{A_1A_2:C|B} = \hat{H}_{A_2:C|B} + \hat{H}_{A_1:C|A_2B}\label{eq:chain_rule}.
\end{equation}
The chain rule can be easily verified from the definition of the conditional mutual spectrum. While any deformation of the subsystem can be expressed as a linear combination of the chain rule, we define three elementary deformation moves for the clarity of the exposition.

First example is an {\it isolation move}. Goal of the isolation move is to deform the boundary between the target party and the reference party so that $X$ can be sufficiently separated from the reference party, see FIG.\ref{fig:isolation}.
\begin{figure}[h!]
\includegraphics[width=0.48\textwidth]{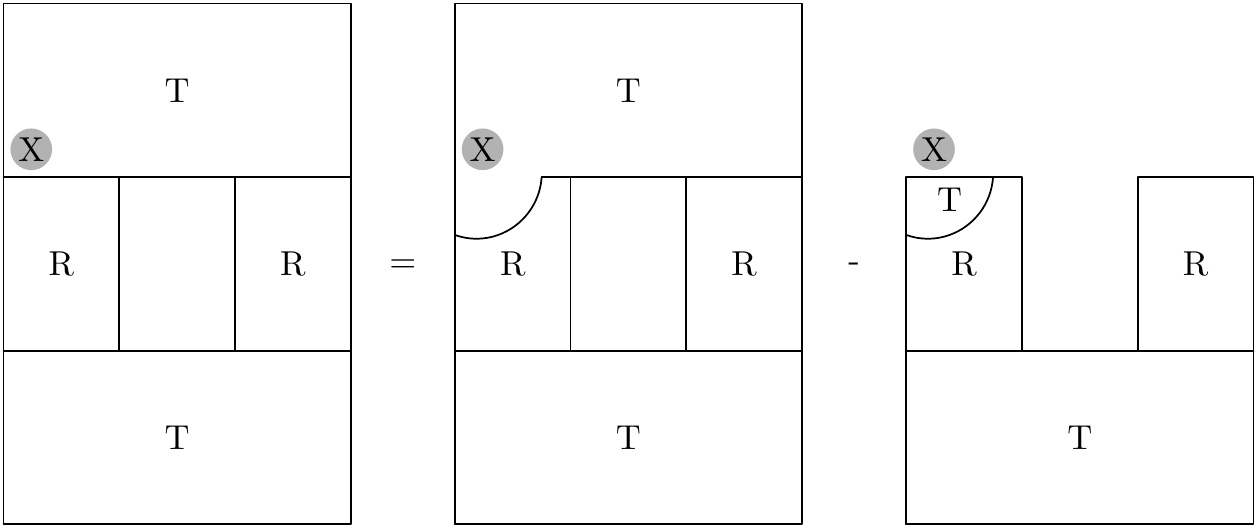}
\caption{After applying the isolation move, the conditional mutual spectrum is deformed in such a way that i) for the new conditional mutual spectrum, $X$ is sufficiently far away from the reference party ii) the difference is a conditional mutual spectrum with small conditional mutual information. \label{fig:isolation}}
\end{figure}
We also define a {\it separation move}. The purpose of the separation move is to deform the target party so that $X$ is sufficiently separated from the target party, See FIG.\ref{fig:separation}.
\begin{figure}[h!]
\includegraphics[width=0.48 \textwidth]{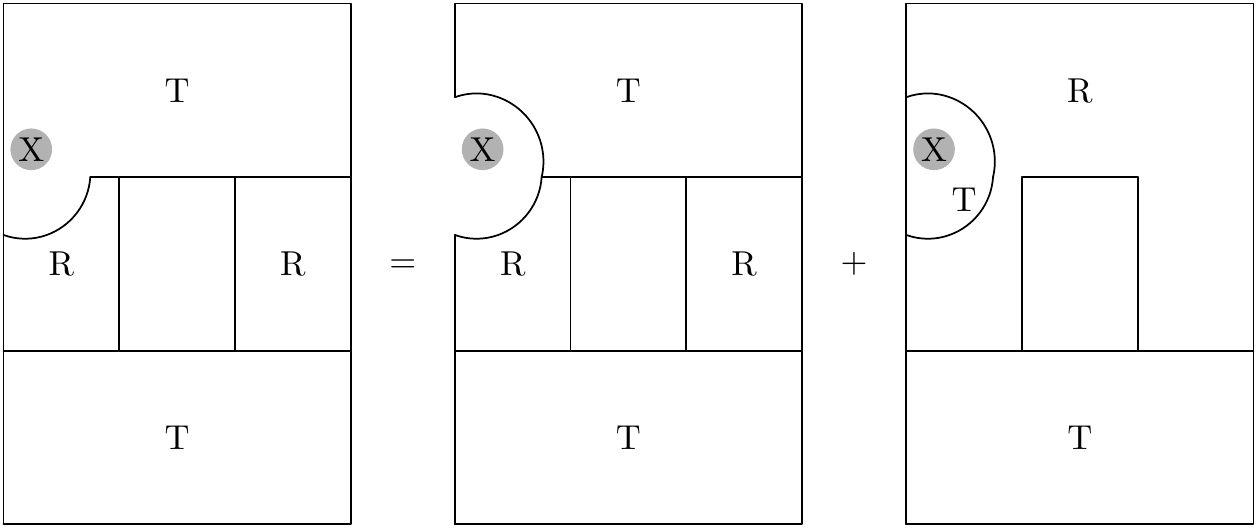}
\caption{After applying the separation move, the conditional mutual spectrum is deformed in such a way that i) for the new conditional mutual spectrum, $X$ is sufficiently far away from both the reference and target party ii) the difference is a conditional mutual spectrum with small conditional mutual information. \label{fig:separation}}
\end{figure}

By first applying the isolation move and then the separation move, one can always deform the configuration to be distance $O(l)$ away from $X$. The correction from the deformation procedure is of the form $\Tr(\rho_{A_iB_iC_iX} \hat{H}_{A_i:C_i|B_i}X)$ with $I(A_i:C_i|B_i)=o(1)$. To bound these terms, we introduce an {\it absorption move}, see FIG.\ref{fig:absorption}.
\begin{figure}[h!]
\includegraphics[width=0.48 \textwidth]{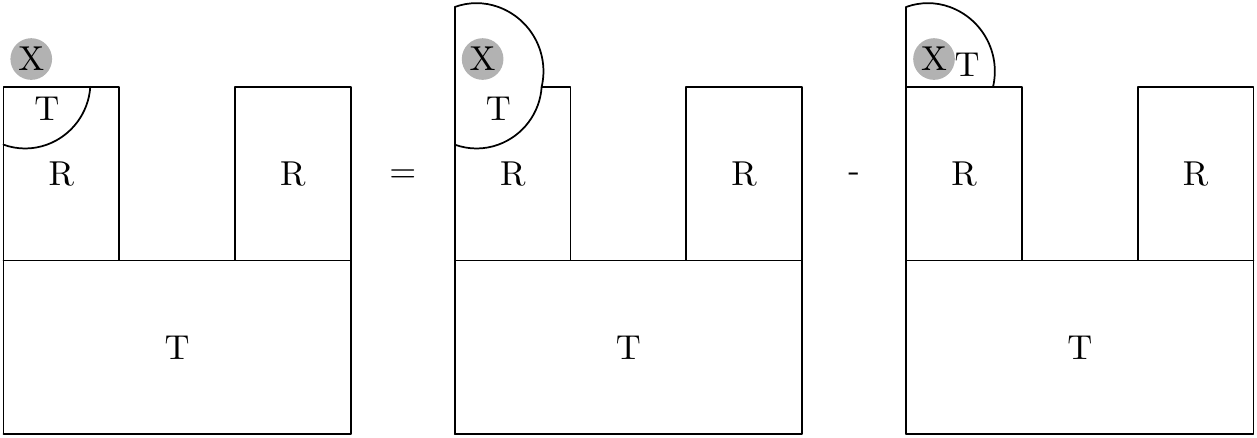}
\caption{Goal of the absorption move is to change the correction terms into conditional mutual spectrum $\hat{H}_{A_i:C_i|B_i}$ such that i) support of $X$ is contained in either $A_i$ or $C_i$ ii) $I(A_i:C_i|B_i)$ is small. \label{fig:absorption}}
\end{figure}

After applying the absorption move, corrections from the deformation move can be expressed as a sum of terms of the form $\Tr(\rho_{A_iB_iC_i} \hat{H}_{A_i:C_i|B_i} X)$ with $X \subset A_i$. These terms can be bounded using Lemma \ref{lemma:lemma1}.

To summarize, given a topologically nontrivial configuration, $\mathcal{C}(\hat{H}_{A:C|B},X)$ can be expressed as $\mathcal{C}(\hat{H}_{A':C'|B'},X)$ with $d(A'B'C', X) = O(l)$ and the correction terms that can be expressed as a sum of $\epsilon_{A_i} \|X \|$ and $\epsilon_{A_i:B_i} \|X \|$. Assuming i) $X$ is localized in one of the original subsystems $A,B,C,(ABC)^c$ ii) each of the subsystems are sufficiently large, the correction terms vanish in the $l \to \infty$ limit. One may be tempted to think that $\mathcal{C}(\hat{H}_{A':C'|B'},X)$ vanishes in the $l \to \infty$ limit as well, since correlation decays exponentially in the ground state of a gapped system.\cite{Hastings2006,Nachtergaele2006} While this speculation turns out to be correct, we emphasize that a slight modification of the exponential clustering theorem is necessary.

\subsection{Modified form of exponential clustering theorem}
Before we explain the details of our analysis, we would like to present a technical background about the subject. Exponential clustering theorem states that
\begin{equation}
|\mathcal{C}(O_A,O_B)| \leq c \|O_A \| \|O_B \| \min(|A|, |B|) e^{-\frac{d(A,B)}{\xi}}
\end{equation}
for two spatially separated operator $O_A$ and $O_B$, provided there is a gapped parent hamiltonian that consists of sum of geometrically local bounded-norm terms.\cite{Hastings2006,Nachtergaele2006} Since the spectrum of $\hat{H}_A$ is formally unbounded, one cannot directly apply exponential clustering theorem. We circumvent this problem by regularizing the entanglement spectrum and bounding the error from the regularization procedure.

\begin{defi}
Regularized entanglement spectrum $\hat{H}_{A}^{\Lambda}$ with a cutoff $\Lambda$ is
\begin{equation}
\hat{H}_A^{\Lambda} = -\sum_{p \geq 1/\Lambda} \log p_i \ket{i} \bra{i}.
\end{equation}
\end{defi}
Simple consequence of this construction is that $l_{\infty}$ norm is bounded, i.e. $\| \hat{H}_A^{\Lambda}\| \leq \log \Lambda$. Correction from the regularization can be bounded using the following lemma.
\begin{lem}
\begin{equation}
\Tr (\rho_{AB}\Delta_A^{\Lambda} O_B ) \leq \|O_B \| \frac{\log \Lambda}{\Lambda} d_A
\end{equation}
for $\Lambda \geq 2$, where $\Delta_{A}^{\Lambda} = \hat{H}_A - \hat{H}_A^{\Lambda}$.
\end{lem}
\begin{proof}
Purify $\rho_{AB}$ to $\ket{\psi}_{ABC}$. Rewrite the formula as $\Tr (\rho_{AB}\Delta_A^{\Lambda} O_B )  = \bra{\psi}_{ABC} \Delta_A^{\Lambda} O_B \ket{\psi}_{ABC}$. Note that $\ket{\psi}_{ABC}$ admits a Schmidt decomposition $\ket{\psi}_{ABC} = \sum_i \sqrt{p_i} \ket{i}_A \ket{i}_{BC}$, where $\rho_A = \sum_i p_i \ket{i}_A \bra{i}_A$. This in turn can be expressed as
\begin{equation}
\sum_{p_i \leq 1/\Lambda} -p_i \log p_i \bra{i}_{BC} O_B \ket{i}_{BC}.
\end{equation}
Using $-p_i \log p_i \leq \frac{1}{\Lambda} \log \Lambda$ and $|\bra{i} O_B \ket{i}| \leq \|O_B \|$, one can complete the proof.
\end{proof}

\section{Physical interpretation\label{section:interpretation}}
Setting $\Lambda= d_{ABC} e^{O(l)/ \xi}$, we arrive at the following conclusion.
\begin{equation}
|\mathcal{C}(\hat{H}_{A:C|B},X)| \leq \| X \|(\epsilon_1(l) + \epsilon_2(l))l^2,\label{eq:correlation_bound}
\end{equation}
where $\epsilon_1$ represents a deviation from the ideal area law, and $\epsilon_2$  represents an error from the long range correlation. As $l \to \infty$, conditional mutual spectrum has vanishing correlation with any local operator, provided that i) $X$ is supported on one of $A,B,C,$ or $(ABC)^c$ ii) both $\epsilon_1$ and $\epsilon_2$ decays sufficiently fast. In $l \to \infty$ limit, we have
\begin{equation}
\langle \hat{H}_{A:C|B} X\rangle = I(A:C|B) \langle X \rangle.\label{eq:main_result}
\end{equation}

We conclude that operator $\hat{H}_{A:C|B}$ is topological, in a sense that i) it has vanishing correlation with any operator that is localized in one of the subsystems ii) its eigenvalues contain information about the phase. A set of assumptions to conclude so was that i) correlation decays exponentially ii) the extensive terms of the entanglement entropy cancel out each other iii) the deformation procedure separating $X$ from $ABC$ does not change the topology of the configuration.

We emphasize that the derivation of our result is not necessarily restricted to a pure state. At finite temperature, entanglement entropy obtains volume contributions, but one may be able to show that those contributions can be canceled out as well. In particular, we expect these conditions to be met for quantum many-body systems at sufficiently high temperature.

In the large volume limit, it seems the local contribution of the reduced density matrices cancel out each other, at least when $I(A:C|B)=o(\frac{1}{l^2})$. We do not have a definitive proof for this statement, but we argue as follows. If $\hat{H}_{A:C|B}$ contains a localized term, one could have chosen $X$ to be an operator supported nearby so as to have a large correlation with the local term. Such terms will violate Eq.\ref{eq:main_result}. Our result suggests a decomposition of the entanglement spectrum into i) terms that can be canceled out by a  suitable choice of subsystems ii) terms that cannot be canceled out and have small correlation with almost any local operators. It would be interesting if the terms of the first kind can be shown to be quasilocal.

\section{Conclusion}
We have presented a general argument as to why certain linear combination entanglement spectrum allows cancelation of local degrees of freedom, owing in part to a recently discovered information-theoretic inequality. While our formulation is not as precise as the ones described by the variational wavefunction,\cite{Cirac2011,Sterdyniak2011,Dubail2012} it has an advantage of being applicable to a more general class of quantum states. Indeed, we have only used an approximate form of area law and the exponential clustering theorem, which are strongly believed to be a generic property of a gapped phase.

It would be interesting if the approximate conditional independence can be shown to hold in other systems. There are evidences suggesting that models based on BF theory should satisfy such condition\cite{Castelnovo2008}, yet no studies have been performed for exotic models in 3D such as Haah's code.\cite{Haah2011} As for the finite temperature states, approximate conditional independence is one of the key ideas of quantum belief propagation(QBP) algorithm.\cite{Hastings2007c} Success of the QBP indicates that our result may be applicable to finite temperature quantum states as well.\cite{Bilgin2009}

On the other hand, we wish to find a deeper insight as to why conditional independence arises in these systems. In particular, exactly solvable models which satisfy exact conditional independence can be thought as a fixed point of some renormalization-group procedure.\cite{Koenig2008} Does conditional mutual information of topologically trivial configurations monotonically decrease under such flow?

We conclude with a remark that our correlation bound cannot be applied to operators that are supported on more than one of the subsystems. Ability to bound correlation of such form can be used for showing perturbative stability of topological entanglement entropy, but that shall be published elsewhere.\cite{Kim2012}

{\it Acknowledgements---} This research was supported in part by NSF under Grant No. PHY-0803371, by ARO Grant No. W911NF-
09-1-0442, and DOE Grant No. DE-FG03-92-ER40701. I would like to thank Alexei Kitaev for helpful discussions.
\bibliography{bib}

\begin{thebibliography}{28}
\expandafter\ifx\csname natexlab\endcsname\relax\def\natexlab#1{#1}\fi
\expandafter\ifx\csname bibnamefont\endcsname\relax
  \def\bibnamefont#1{#1}\fi
\expandafter\ifx\csname bibfnamefont\endcsname\relax
  \def\bibfnamefont#1{#1}\fi
\expandafter\ifx\csname citenamefont\endcsname\relax
  \def\citenamefont#1{#1}\fi
\expandafter\ifx\csname url\endcsname\relax
  \def\url#1{\texttt{#1}}\fi
\expandafter\ifx\csname urlprefix\endcsname\relax\def\urlprefix{URL }\fi
\providecommand{\bibinfo}[2]{#2}
\providecommand{\eprint}[2][]{\url{#2}}

\bibitem[{\citenamefont{Eisert et~al.}(2008)\citenamefont{Eisert, Cramer, and
  Plenio}}]{Eisert2008}
\bibinfo{author}{\bibfnamefont{J.}~\bibnamefont{Eisert}},
  \bibinfo{author}{\bibfnamefont{M.}~\bibnamefont{Cramer}}, \bibnamefont{and}
  \bibinfo{author}{\bibfnamefont{M.~B.} \bibnamefont{Plenio}},
  \bibinfo{journal}{Rev. Mod. Phys.} \textbf{\bibinfo{volume}{82}},
  \bibinfo{pages}{277} (\bibinfo{year}{2008}), \eprint{0808.3773}.

\bibitem[{\citenamefont{Hastings}(2007{\natexlab{a}})}]{Hastings2007a}
\bibinfo{author}{\bibfnamefont{M.~B.} \bibnamefont{Hastings}},
  \bibinfo{journal}{JSTAT, P} \textbf{\bibinfo{volume}{08024}}
  (\bibinfo{year}{2007}{\natexlab{a}}), \eprint{0705.2024}.

\bibitem[{\citenamefont{Levin and Wen}(2005)}]{Levin2005}
\bibinfo{author}{\bibfnamefont{M.~A.} \bibnamefont{Levin}} \bibnamefont{and}
  \bibinfo{author}{\bibfnamefont{X.-G.} \bibnamefont{Wen}},
  \bibinfo{journal}{Phys.Rev. B} \textbf{\bibinfo{volume}{71}},
  \bibinfo{pages}{045110} (\bibinfo{year}{2005}), \eprint{cond-mat/0404617}.

\bibitem[{\citenamefont{Verstraete et~al.}(2006)\citenamefont{Verstraete, Wolf,
  Perez-Garcia, and Cirac}}]{Verstraete2006}
\bibinfo{author}{\bibfnamefont{F.}~\bibnamefont{Verstraete}},
  \bibinfo{author}{\bibfnamefont{M.~M.} \bibnamefont{Wolf}},
  \bibinfo{author}{\bibfnamefont{D.}~\bibnamefont{Perez-Garcia}},
  \bibnamefont{and} \bibinfo{author}{\bibfnamefont{J.~I.} \bibnamefont{Cirac}},
  \bibinfo{journal}{Phys. Rev. Lett.} \textbf{\bibinfo{volume}{96}},
  \bibinfo{pages}{220601.} (\bibinfo{year}{2006}), \eprint{quant-ph/0601075}.

\bibitem[{\citenamefont{Kitaev and Preskill}(2006)}]{Kitaev2006}
\bibinfo{author}{\bibfnamefont{A.}~\bibnamefont{Kitaev}} \bibnamefont{and}
  \bibinfo{author}{\bibfnamefont{J.}~\bibnamefont{Preskill}},
  \bibinfo{journal}{Phys.Rev.Lett.} \textbf{\bibinfo{volume}{96}},
  \bibinfo{pages}{110404} (\bibinfo{year}{2006}), \eprint{hep-th/0510092}.

\bibitem[{\citenamefont{Levin and Wen}(2006)}]{Levin2006}
\bibinfo{author}{\bibfnamefont{M.}~\bibnamefont{Levin}} \bibnamefont{and}
  \bibinfo{author}{\bibfnamefont{X.-G.} \bibnamefont{Wen}},
  \bibinfo{journal}{Phys. Rev. Lett.} \textbf{\bibinfo{volume}{96}},
  \bibinfo{pages}{110405} (\bibinfo{year}{2006}), \eprint{cond-mat/0510613}.

\bibitem[{\citenamefont{Li and Haldane}(2008)}]{Li2008}
\bibinfo{author}{\bibfnamefont{H.}~\bibnamefont{Li}} \bibnamefont{and}
  \bibinfo{author}{\bibfnamefont{F.~D.~M.} \bibnamefont{Haldane}},
  \bibinfo{journal}{Phys. Rev. Lett.} \textbf{\bibinfo{volume}{101}},
  \bibinfo{pages}{010504} (\bibinfo{year}{2008}), \eprint{0805.0332}.

\bibitem[{\citenamefont{Thomale et~al.}(2010)\citenamefont{Thomale, Sterdyniak,
  Regnault, and Bernevig}}]{Thomale2010}
\bibinfo{author}{\bibfnamefont{R.}~\bibnamefont{Thomale}},
  \bibinfo{author}{\bibfnamefont{A.}~\bibnamefont{Sterdyniak}},
  \bibinfo{author}{\bibfnamefont{N.}~\bibnamefont{Regnault}}, \bibnamefont{and}
  \bibinfo{author}{\bibfnamefont{B.~A.} \bibnamefont{Bernevig}},
  \bibinfo{journal}{Phys. Rev. Lett.} \textbf{\bibinfo{volume}{104}},
  \bibinfo{pages}{180502} (\bibinfo{year}{2010}).

\bibitem[{\citenamefont{Sterdyniak et~al.}(2011)\citenamefont{Sterdyniak,
  Chandran, Regnault, Bernevig, and Bonderson}}]{Sterdyniak2011}
\bibinfo{author}{\bibfnamefont{A.}~\bibnamefont{Sterdyniak}},
  \bibinfo{author}{\bibfnamefont{A.}~\bibnamefont{Chandran}},
  \bibinfo{author}{\bibfnamefont{N.}~\bibnamefont{Regnault}},
  \bibinfo{author}{\bibfnamefont{B.~A.} \bibnamefont{Bernevig}},
  \bibnamefont{and}
  \bibinfo{author}{\bibfnamefont{P.}~\bibnamefont{Bonderson}},
  \bibinfo{journal}{Phys. Rev. B} \textbf{\bibinfo{volume}{85}},
  \bibinfo{pages}{125308} (\bibinfo{year}{2011}), \eprint{1111.2810}.

\bibitem[{\citenamefont{Dubail et~al.}(2011)\citenamefont{Dubail, Read, and
  Rezayi}}]{Dubail2011}
\bibinfo{author}{\bibfnamefont{J.}~\bibnamefont{Dubail}},
  \bibinfo{author}{\bibfnamefont{N.}~\bibnamefont{Read}}, \bibnamefont{and}
  \bibinfo{author}{\bibfnamefont{E.~H.} \bibnamefont{Rezayi}},
  \bibinfo{journal}{Phys. Rev. B} \textbf{\bibinfo{volume}{85}},
  \bibinfo{pages}{115321} (\bibinfo{year}{2011}), \eprint{1111.2811}.

\bibitem[{\citenamefont{Rodr\'iguez et~al.}(2012)\citenamefont{Rodr\'iguez,
  Simon, and Slingerland}}]{Rodriguez2012}
\bibinfo{author}{\bibfnamefont{I.~D.} \bibnamefont{Rodr\'iguez}},
  \bibinfo{author}{\bibfnamefont{S.~H.} \bibnamefont{Simon}}, \bibnamefont{and}
  \bibinfo{author}{\bibfnamefont{J.~K.} \bibnamefont{Slingerland}},
  \bibinfo{journal}{Phys. Rev. Lett.} \textbf{\bibinfo{volume}{108}},
  \bibinfo{pages}{256806} (\bibinfo{year}{2012}).

\bibitem[{\citenamefont{Cirac et~al.}(2011)\citenamefont{Cirac, Poilblanc,
  Schuch, and Verstraete}}]{Cirac2011}
\bibinfo{author}{\bibfnamefont{J.~I.} \bibnamefont{Cirac}},
  \bibinfo{author}{\bibfnamefont{D.}~\bibnamefont{Poilblanc}},
  \bibinfo{author}{\bibfnamefont{N.}~\bibnamefont{Schuch}}, \bibnamefont{and}
  \bibinfo{author}{\bibfnamefont{F.}~\bibnamefont{Verstraete}},
  \bibinfo{journal}{Phys. Rev. B} \textbf{\bibinfo{volume}{83}},
  \bibinfo{pages}{245134} (\bibinfo{year}{2011}), \eprint{1103.3427}.

\bibitem[{\citenamefont{Dubail et~al.}(2012)\citenamefont{Dubail, Read, and
  Rezayi}}]{Dubail2012}
\bibinfo{author}{\bibfnamefont{J.}~\bibnamefont{Dubail}},
  \bibinfo{author}{\bibfnamefont{N.}~\bibnamefont{Read}}, \bibnamefont{and}
  \bibinfo{author}{\bibfnamefont{E.~H.} \bibnamefont{Rezayi}},
  \bibinfo{journal}{Phys. Rev. B} \textbf{\bibinfo{volume}{86}},
  \bibinfo{pages}{245310} (\bibinfo{year}{2012}).

\bibitem[{\citenamefont{Poulin and Hastings}(2010)}]{Poulin2010a}
\bibinfo{author}{\bibfnamefont{D.}~\bibnamefont{Poulin}} \bibnamefont{and}
  \bibinfo{author}{\bibfnamefont{M.~B.} \bibnamefont{Hastings}},
  \bibinfo{journal}{Phys. Rev. Lett.} \textbf{\bibinfo{volume}{106}},
  \bibinfo{pages}{080403} (\bibinfo{year}{2010}), \eprint{1012.2050}.

\bibitem[{\citenamefont{Lieb and Ruskai}(1973)}]{Lieb1973}
\bibinfo{author}{\bibfnamefont{E.~H.} \bibnamefont{Lieb}} \bibnamefont{and}
  \bibinfo{author}{\bibfnamefont{M.~B.} \bibnamefont{Ruskai}},
  \bibinfo{journal}{J. Math. Phys.} \textbf{\bibinfo{volume}{14}},
  \bibinfo{pages}{1938} (\bibinfo{year}{1973}).

\bibitem[{\citenamefont{Petz}(2003)}]{Petz2003}
\bibinfo{author}{\bibfnamefont{D.}~\bibnamefont{Petz}}, \bibinfo{journal}{Rev.
  Math. Phys.} \textbf{\bibinfo{volume}{15}}, \bibinfo{pages}{79}
  (\bibinfo{year}{2003}), \eprint{quant-ph/0209053}.

\bibitem[{\citenamefont{Hayden et~al.}(2004)\citenamefont{Hayden, Jozsa, Petz,
  and Winter}}]{Hayden2004}
\bibinfo{author}{\bibfnamefont{P.}~\bibnamefont{Hayden}},
  \bibinfo{author}{\bibfnamefont{R.}~\bibnamefont{Jozsa}},
  \bibinfo{author}{\bibfnamefont{D.}~\bibnamefont{Petz}}, \bibnamefont{and}
  \bibinfo{author}{\bibfnamefont{A.}~\bibnamefont{Winter}},
  \bibinfo{journal}{Commun. Math. Phys.} \textbf{\bibinfo{volume}{246}},
  \bibinfo{pages}{359} (\bibinfo{year}{2004}), \eprint{quant-ph/0304007}.

\bibitem[{\citenamefont{Leifer and Poulin}(2007)}]{Leifer2007}
\bibinfo{author}{\bibfnamefont{M.}~\bibnamefont{Leifer}} \bibnamefont{and}
  \bibinfo{author}{\bibfnamefont{D.}~\bibnamefont{Poulin}},
  \bibinfo{journal}{Ann. Phys.} \textbf{\bibinfo{volume}{323}},
  \bibinfo{pages}{1899} (\bibinfo{year}{2007}), \eprint{0708.1337}.

\bibitem[{\citenamefont{Kim}(2012{\natexlab{a}})}]{Kim2012a}
\bibinfo{author}{\bibfnamefont{I.~H.} \bibnamefont{Kim}}, \bibinfo{journal}{J.
  Math. Phys.} \textbf{\bibinfo{volume}{53}}, \bibinfo{pages}{122204}
  (\bibinfo{year}{2012}{\natexlab{a}}), \eprint{1210.5190}.

\bibitem[{\citenamefont{Hastings and Koma}(2006)}]{Hastings2006}
\bibinfo{author}{\bibfnamefont{M.~B.} \bibnamefont{Hastings}} \bibnamefont{and}
  \bibinfo{author}{\bibfnamefont{T.}~\bibnamefont{Koma}},
  \bibinfo{journal}{Commun.Math.Phys.} \textbf{\bibinfo{volume}{265}},
  \bibinfo{pages}{781} (\bibinfo{year}{2006}), \eprint{math-ph/0507008}.

\bibitem[{\citenamefont{Nachtergaele and Sims}(2006)}]{Nachtergaele2006}
\bibinfo{author}{\bibfnamefont{B.}~\bibnamefont{Nachtergaele}}
  \bibnamefont{and} \bibinfo{author}{\bibfnamefont{R.}~\bibnamefont{Sims}},
  \bibinfo{journal}{Commun. Math. Phys.} \textbf{\bibinfo{volume}{265}},
  \bibinfo{pages}{119} (\bibinfo{year}{2006}), \eprint{math-ph/0506030}.

\bibitem[{\citenamefont{Castelnovo and Chamon}(2008)}]{Castelnovo2008}
\bibinfo{author}{\bibfnamefont{C.}~\bibnamefont{Castelnovo}} \bibnamefont{and}
  \bibinfo{author}{\bibfnamefont{C.}~\bibnamefont{Chamon}},
  \bibinfo{journal}{Phys. Rev. B} \textbf{\bibinfo{volume}{78}},
  \bibinfo{pages}{155120} (\bibinfo{year}{2008}), \eprint{0804.3591}.

\bibitem[{\citenamefont{Haah}(2011)}]{Haah2011}
\bibinfo{author}{\bibfnamefont{J.}~\bibnamefont{Haah}}, \bibinfo{journal}{Phys.
  Rev. A} \textbf{\bibinfo{volume}{83}}, \bibinfo{pages}{042330}
  (\bibinfo{year}{2011}), \eprint{1101.1962}.

\bibitem[{\citenamefont{Hastings}(2007{\natexlab{b}})}]{Hastings2007c}
\bibinfo{author}{\bibfnamefont{M.~B.} \bibnamefont{Hastings}},
  \bibinfo{journal}{Phys. Rev. B Rapids} \textbf{\bibinfo{volume}{76}},
  \bibinfo{pages}{201102} (\bibinfo{year}{2007}{\natexlab{b}}),
  \eprint{0706.4094}.

\bibitem[{\citenamefont{Bilgin and Poulin}(2009)}]{Bilgin2009}
\bibinfo{author}{\bibfnamefont{E.}~\bibnamefont{Bilgin}} \bibnamefont{and}
  \bibinfo{author}{\bibfnamefont{D.}~\bibnamefont{Poulin}},
  \bibinfo{journal}{Phys. Rev. B} \textbf{\bibinfo{volume}{81}},
  \bibinfo{pages}{054106} (\bibinfo{year}{2009}), \eprint{0910.2299}.

\bibitem[{\citenamefont{Koenig et~al.}(2008)\citenamefont{Koenig, Reichardt,
  and Vidal}}]{Koenig2008}
\bibinfo{author}{\bibfnamefont{R.}~\bibnamefont{Koenig}},
  \bibinfo{author}{\bibfnamefont{B.~W.} \bibnamefont{Reichardt}},
  \bibnamefont{and} \bibinfo{author}{\bibfnamefont{G.}~\bibnamefont{Vidal}},
  \bibinfo{journal}{Phys. Rev. B} \textbf{\bibinfo{volume}{79}},
  \bibinfo{pages}{195123} (\bibinfo{year}{2008}), \eprint{0806.4583}.

\bibitem[{\citenamefont{Kim}(2012{\natexlab{b}})}]{Kim2012}
\bibinfo{author}{\bibfnamefont{I.~H.} \bibnamefont{Kim}},
  \bibinfo{journal}{Phys. Rev. B} \textbf{\bibinfo{volume}{86}},
  \bibinfo{pages}{245116} (\bibinfo{year}{2012}{\natexlab{b}}).

\bibitem[{\citenamefont{Wolf et~al.}(2007)\citenamefont{Wolf, Verstraete,
  Hastings, and Cirac}}]{Wolf2007}
\bibinfo{author}{\bibfnamefont{M.~M.} \bibnamefont{Wolf}},
  \bibinfo{author}{\bibfnamefont{F.}~\bibnamefont{Verstraete}},
  \bibinfo{author}{\bibfnamefont{M.~B.} \bibnamefont{Hastings}},
  \bibnamefont{and} \bibinfo{author}{\bibfnamefont{J.~I.} \bibnamefont{Cirac}},
  \bibinfo{journal}{Phys. Rev. Lett.} \textbf{\bibinfo{volume}{100}},
  \bibinfo{pages}{070502} (\bibinfo{year}{2007}), \eprint{0704.3906}.

\end{thebibliography}

\end{document}